\documentclass[submission,copyright,creativecommons]{eptcs}
\providecommand{\event}{FICS 2024} 

\usepackage{iftex}
\ifpdf
  \usepackage{underscore}         
  \usepackage[T1]{fontenc}        
\else
  \usepackage{breakurl}           
\fi

\usepackage{amsmath,amssymb,amsfonts}
\usepackage{amsthm}

\newtheorem*{main-theorem}{Main Theorem}
\newtheorem{theorem}{Theorem}
\newtheorem{lemma}[theorem]{Lemma}
\newtheorem{proposition}[theorem]{Proposition}
\newtheorem*{corollary}{Corollary}

\theoremstyle{plain}
\newtheorem*{conjecture}{Conjecture}
\newtheorem*{problem}{Problem}
\theoremstyle{definition}

\usepackage{tikz}

\newcommand{\verifier}{\mathsf{V}}
\newcommand{\refuter}{\mathsf{R}}

\newcommand{\restrictedto}{\upharpoonright}
\newcommand{\powerset}[1]{\mathcal{P}({#1})}
\newcommand{\tuple}[1]{{\langle #1 \rangle}}
\newcommand{\bd}{\mathrm{bd}}
\newcommand{\st}{\mathrm{pos}}
\newcommand{\prezero}{\mathrm{pre}_0}
\newcommand{\preone}{\mathrm{pre}_1}
\newcommand{\pretwo}{\mathrm{pre}_2}
\newcommand{\midzero}{\mathrm{mid}_0}
\newcommand{\nextzero}{\mathrm{nxt}_0}
\newcommand{\nextone}{\mathrm{nxt}_1}
\newcommand{\nexttwo}{\mathrm{nxt}_2}

\title{The $\mu$-calculus' Alternation Hierarchy is Strict\\over Non-Trivial Fusion Logics}
\author{Leonardo Pacheco\thanks{I would like to thank Thibaut Kouptchinky for comments and proof-reading. I would also like to thank the reviewers for their comments. This research was partially funded by the FWF grant TAI-797.}
\institute{TU Wien\\ Vienna, Austria}
\email{leonardo.pacheco@tuwien.ac.at}}
\def\titlerunning{The $\mu$-calculus' Alternation Hierarchy is Strict over Non-Trivial Fusion Logics}
\def\authorrunning{L. Pacheco}
\begin{document}
\maketitle

\begin{abstract}
The modal $\mu$-calculus is obtained by adding least and greatest fixed-point operators to modal logic.
Its alternation hierarchy classifies the $\mu$-formulas by their alternation depth: a measure of the codependence of their least and greatest fixed-point operators.
The $\mu$-calculus' alternation hierarchy is strict over the class of all Kripke frames: for all $n$, there is a $\mu$-formula with alternation depth $n+1$ which is not equivalent to any formula with alternation depth $n$.
This does not always happen if we restrict the semantics.
For example, every $\mu$-formula is equivalent to a formula without fixed-point operators over $\mathsf{S5}$ frames.
We show that the multimodal $\mu$-calculus' alternation hierarchy is strict over non-trivial fusions of modal logics.
We also comment on two examples of multimodal logics where the $\mu$-calculus collapses to modal logic.
\end{abstract}

\section{Introduction}
The modal $\mu$-calculus is obtained by adding least and greatest fixed-point operators to modal logic.
One measure of complexity for $\mu$-formulas is their alternation depth, which measures the codependence of least and greatest fixed-point operators.
Bradfield \cite{bradfield1998simplifying} showed that the $\mu$-calculus' alternation hierarchy is strict: for all $n\in\mathbb{N}$, there is a formula with alternation depth $n+1$ which is not equivalent over unimodal frames to any formula with alternation depth $n$.
On the other hand, Alberucci and Facchini \cite{alberucci2009modal} proved that, over $\mathsf{S5}$ frames, every $\mu$-formula is equivalent to a formula without fixed-point operators.
See Chapter 2 of \cite{pacheco2023exploring} for a survey on the $\mu$-calculus' alternation hierarchy over various classes of frames.

Let $\mathsf{L}_0$ and $\mathsf{L}_1$ be modal logics with disjoint signatures.
The fusion $\mathsf{L}_0\otimes \mathsf{L}_1$ is the smallest modal logic containing both $\mathsf{L}_0$ and $\mathsf{L}_1$.
If $\mathsf{L}_0$ and $\mathsf{L}_1$ are respectively characterized by the Kripke frames in $\mathsf{F}_0$ and $\mathsf{F}_1$, then the fusion $\mathsf{L}_0 \otimes \mathsf{L}_1$ is characterized by frames which are in $\mathsf{F}_i$ when restricted to the signature of $\mathsf{L}_i$, for $i= 0,1$.
Fusion logics are commonly used for multi-agent epistemic logics and on the specification of computer systems.
We show that, over fusions of non-trivial classes of frames, the $\mu$-calculus' alternation hierarchy is strict.
Our proof is based on work of Bradfield \cite{bradfield1998simplifying} and Alberucci \cite{alberucci2002strictness}.

Let $\mathsf{F}$ be a class of unimodal Kripke frames.
We say $\circ\leftarrow\circ\to\circ$ is a subframe of $\mathsf{F}$ iff there is some frame $F=\tuple{W,R}\in\mathsf{F}$ with pairwise different $w_0,w_1,w_2\in W$ such that $w_0 R w_1$ and $w_0 R w_2$.
We analogously define $\circ\to\circ\to\circ$ is a subframe of $\mathsf{F}$ and $\circ\to\circ$ is a subframe of $\mathsf{F}$.
We will define multimodal versions $W_n$ of the winning region formulas $W_n'$ to prove:
\begin{main-theorem}
    Let $\mathsf{F}_0$, $\mathsf{F}_1$, and $\mathsf{F}_2$ be classes of unimodal Kripke frames closed under isomorphic copies and disjoint unions.
    If
    \begin{enumerate}
        \item $\circ\leftarrow\circ\to\circ$ is a subframe of $\mathsf{F}_0$ and $\circ\to\circ$ a subframe of $\mathsf{F}_1$; or
        \item $\circ\to\circ\to\circ$ is a subframe of $\mathsf{F}_0$ and $\circ\to\circ$ a subframe of $\mathsf{F}_1$;
    \end{enumerate}
    then the $\mu$-calculus' alternation hierarchy is strict over $\mathsf{F}_0\otimes\mathsf{F}_1$.
    If
    \begin{enumerate}
        \item[3.] $\circ\to\circ$ is a subframe of $\mathsf{F}_0$, $\mathsf{F}_1$, and $\mathsf{F}_2$;
    \end{enumerate}
    then the $\mu$-calculus' alternation hierarchy is strict over $\mathsf{F}_0\otimes\mathsf{F}_1\otimes\mathsf{F}_2$.
\end{main-theorem}
\begin{corollary}
    Let $\{\mathsf{L}_0,\mathsf{L}_1\} \subseteq \{\mathsf{K}, \mathsf{K4}, \mathsf{S4}, \mathsf{KD45}, \mathsf{S5}, \mathsf{GL}\}$, then the $\mu$-calculus' alternation hierarchy is strict over $\mathsf{L}_0\otimes\mathsf{L}_1$.
\end{corollary}

One proof of the strictness of the $\mu$-calculus makes essential use of parity games \cite{bradfield1998simplifying,alberucci2002strictness}.
In this proof, evaluation games for the $\mu$-calculus are encoded as parity games, parity games are encoded as Kripke models, and formulas defining winning regions for parity games are given as witnesses for the strictness.
For the multimodal case, we need to make changes for both of these.
The encoding of parity games gets more complicated as we cannot just use the graph of the game as the graph of the Kripke model, and need to use copies of frames from both classes along with auxiliary propositional symbols in the encoding.
This also complicates the winning region formulas, which need to take into account these auxiliary propositional symbols.

While the hypotheses of the Main Theorem  looks \emph{ad hoc}, we conjecture that they are optimal.
\begin{conjecture}
    Let $\mathsf{F}_0$ and $\mathsf{F}_1$ be classes of unimodal Kripke frames closed under isomorphic copies and disjoint unions.
    Suppose $\circ\to\circ$ is a subframe of $\mathsf{F}_0$ and $\mathsf{F}_1$.
    Then every $\mu$-formula is equivalent to one with alternation depth $1$ over $\mathsf{F}_0\otimes\mathsf{F}_1$.
\end{conjecture}

As a counterpoint, we comment on two multimodal logics where the $\mu$-calculus collapses to modal logic.
$\mathsf{GLP}$ is a provability logic which contains countably many modal operators; its fixed-point property was proved by Ignatiev \cite{ignatiev1993strong}.
$\mathsf{IS5}$ is an intuitionistic version of $\mathsf{S5}$ which can be thought of as a fragment of a bimodal logic; the $\mu$-calculus' collapse to modal logic over $\mathsf{IS5}$ was proved by Pacheco \cite{pacheco2023game}.

\paragraph{Outline}
In Section \ref{sec::preliminaries}, we review some basic definitions.
In Sections \ref{sec::evaluation-games-as-parity-games}, \ref{sec::winning-region-formulas}, and \ref{sec::strictness}, we give a detailed proof of Item 1 of the Main Theorem: we first show that evaluation games for the $\mu$-calculus are also parity games, then define the formulas $W_n$ and show how parity games can be encoded as multimodal Kripke models and, at last, show that $W_n$ is not equivalent to any formula with lower alternation depth.
In Section \ref{sec::other-items}, we sketch how to modify the proof to show Items 2 and 3 of the Main Theorem.
In Section~\ref{sec::case-studies}, we describe two examples of multimodal logics where the $\mu$-calculus collapses to modal logic.

\section{Preliminaries}
\label{sec::preliminaries}

\paragraph{The $\mu$-calculus}
Fix a set $\mathrm{Prop}$ of propositional symbols, a set $\mathrm{Var}$ of variable symbols, and a non-empty signature $\Lambda$.
The \emph{$\mu$-formulas} are generated by the following grammar:
\[
    \varphi := P \;|\; \neg P \;|\; X \;|\; \varphi\land\varphi \;|\; \varphi \lor \varphi \;|\; \Box_i \varphi_i \;|\; \Diamond_i \varphi \;|\; \mu X.\varphi \;|\; \nu X.\varphi,
\]
where $P\in\mathrm{Prop}$, $X\in\mathrm{Var}$ is a variable symbol, and $i\in\Lambda$.
We write $\eta X.\varphi$ for $\mu X.\varphi$ or $\nu X.\varphi$.
The \emph{set of subformulas} of a formula $\varphi$ is denoted by $\mathrm{Sub}(\varphi)$.

Given a signature $\Lambda$, a \emph{Kripke frame} is a pair $M=\tuple{W, \{R_i\}_{i\in\Lambda}}$ where: $W$ is the set of \emph{possible worlds}; and each $R_i$ is a binary relation on $W$, the \emph{accessibility relations}.
A \emph{Kripke model} is a triple $M=\tuple{W,\{R_i\}_{i\in\Lambda}, V}$ obtained by extending a Kripke frame with a function $V$ from propositional symbols to subsets of $W$; $V$ is called a \emph{valuation function}.
Given a set $A\subseteq W$, the \emph{augmented model} $M[X:=A]$ is obtained by setting $V(X):= A$.
A \emph{pointed Kripke model} is a pair $(M,w)$ consisting of a Kripke model $M$ and a world $w$ of $M$.

Fix a Kripke model $M = \tuple{W,\{R_i\}_{i\in\Lambda},V}$.
Given a $\mu$-formula $\varphi(X)$ with a distinguished variable $X$, let $\Gamma_{\varphi(X)}:\powerset{W}\to \powerset{W}$ be the operator which maps $A\subseteq W$ to $\|\varphi(X)\|^{M[X:=A]}$.
We define the valuation $\|\varphi\|^M$ on $M$ inductively on the structure of $\mu$-formulas:
\begin{center}
\begin{tabular}{ll}
    \textbullet \; $\|P\|^M := V(P)$; & \\
    \textbullet \; $\|X\|^{M[X:=A]} := A$; & \textbullet \; $\|\neg\varphi\|^M := W\setminus\|\varphi\|^M$; \\
    \textbullet \; $\|\varphi\land \psi\|^M := \|\varphi\|^M \cap \|\psi\|^M$; & \textbullet \; $\|\varphi\lor \psi\|^M := \|\varphi\|^M \cup \|\psi\|^M$; \\
    \textbullet \; $\|\Box_i \varphi\|^M := \{w \in W \;|\; \forall v.w R_i v \to v\in \|\varphi\|^M\}$; & \textbullet \; $\|\Diamond_i \varphi\|^M := \{w \in W \;|\; \exists v.w R_i v \land v\in \|\varphi\|^M\}$; \\
    \textbullet \; $\|\mu X.\varphi\|^M$ is the least fixed-point of $\Gamma_{\varphi(X)}$; & \textbullet \; $\|\nu X.\varphi\|^M$ is the greatest fixed-point of $\Gamma_{\varphi(X)}$.
\end{tabular}
\end{center}
Note that the operator $\Gamma_{\varphi(X)}$ is monotone for all formula $\varphi(X)$: if $A\subseteq B\subseteq W$, then $\Gamma_{\varphi(X)}(A)\subseteq\Gamma_{\varphi(X)}(B)$.
By the Knaster--Tarski Theorem, the least and greatest fixed-points of $\Gamma_{\varphi(X)}$ are well-defined.
We say a formula $\varphi$ is \emph{valid on a Kripke model $M$} iff $\varphi$ holds on all worlds of $M$.
We say a formula $\varphi$ is \emph{valid on a Kripke frame $F$} iff $\varphi$ is valid on all Kripke models obtained by adding valuations to $F$.
When convenient, we write $M,w\models\varphi$ for $w\in\|\varphi\|^M$.
See \cite{bradfield2018mucalculus} for more information on the $\mu$-calculus.

\paragraph{Fusions}
Fix $n\in\mathbb{N}$.
A (normal) modal logic is a set of formulas (without fixed-point operators) closed containing all the propositional tautologies and closed under \emph{modus ponens}, necessitation, and substitution.
Let $\{\mathsf{L}_j\}_{j\leq n}$ be a collection of modal logics with pairwise disjoint signatures.
The fusion $\bigotimes_{j\leq n}\mathsf{L}_j$ is the smallest modal logic containing the logics $\{\mathsf{L}_j\}_{j\leq n}$.
Let $\{\mathsf{F}_j\}_{j\leq n}$ be classes of frames with pairwise disjoint signatures $\{\Lambda_j\}_{j\leq n}$.
Put $\Lambda = \bigcup_{j\leq n}\Lambda_j$.
Define $\bigotimes_{j\leq n}\mathsf{F}_j$ as the class of frames $F = \tuple{W, \{R_i\}_{i\in\Lambda}}$ such that $\tuple{W, \{R_i\}_{i\in\Lambda_j}}$ is a frame of $\mathsf{F}_j$ for all $j\leq n$.

Suppose the modal logic $\mathsf{L}_j$ is characterized by the class of frames $\mathsf{F}_j$, for all $j\leq n$.
Then $\bigotimes_{j\leq n}\mathsf{L}_j$ is characterized by $\bigotimes_{j\leq n}\mathsf{F}_j$.
Furthermore, if all the $\mathsf{L}_j$ have the finite model property, then $\bigotimes_{j\leq n}\mathsf{L}_j$ also has the finite model property.
Similarly, if all the $\mathsf{L}_j$ are decidable, so is $\bigotimes_{j\leq n}\mathsf{L}_j$.
On the other hand, fusions do not preserve the complexity of the logics: almost all interesting fusions are $\mathrm{PSPACE}$-hard.
See \cite{kurucz2007combining,carnielli2020combining} for more on fusions of modal logics and other combinations of modal logics.

\paragraph{Alternation Hierarchy}
The $\mu$-calculus’ alternation hierarchy classifies the $\mu$-formulas according to the co-dependence of its least and greatest fixed-point operators.
We define it as follows:
\begin{itemize}
    \item $\Sigma^\mu_0 (= \Pi^\mu_0)$ is the set of all $\mu$-formulas with no fixed-point operators.
    \item $\Sigma^\mu_{n+1}$ is the closure of $\Sigma^\mu_n \cup \Pi^\mu_n$ under
        propositional operators,
        modal operators,
        $\mu X$,
        and the substitution: if $\varphi(X)\in \Sigma^{\mu}_{n+1}$ and $\psi\in \Sigma^{\mu}_{n+1}$ are such that no free variable of $\psi$ becomes bound in $\varphi(\psi)$, then $\varphi(\psi)\in \Sigma^{\mu}_{n+1}$.
    \item $\Pi^\mu_{n+1}$ is the closure of $\Sigma^\mu_n \cup \Pi^\mu_n$ under
        propositional symbols,
        modal operators,
        $\nu X$,
        and the analogous substitution: if $\varphi(X)\in \Pi^{\mu}_{n+1}$ and $\psi\in \Pi^{\mu}_{n+1}$ are such that no free variable of $\psi$ becomes bound in $\varphi(\psi)$, then $\varphi(\psi)\in \Pi^{\mu}_{n+1}$.
\end{itemize}

Let $\mathsf{F}$ be a class of Kripke frames.
The $\mu$-calculus' alternation hierarchy is strict over $\mathsf{F}$ iff, for all $n$, there is a formula in $\Sigma^\mu_{n+1}\cup\Pi^\mu_{n+1}$ which is not equivalent to any formula in $\Sigma^\mu_{n}\cup\Pi^\mu_{n}$ over $\mathsf{F}$.
The $\mu$-calculus collapses to modal logic over $\mathsf{F}$ iff every $\mu$-formula is equivalent to a formula without fixed-point operators over $\mathsf{F}$.

\paragraph{Game Semantics}
The $\mu$-calculus also has an equivalent game semantics.
Fix a $\mu$-formula $\varphi$, a Kripke model $M = \tuple{W,\{R_i\}_{i\in\Lambda},V}$, and a world $w$.
For notational simplicity, we suppose each variable occurring in $\varphi$ has only one occurrence and is bound by some fixed-point operator.\footnote{This statement is not problematic as we are interested in metamathematical properties of the $\mu$-calculus. More care is needed when one is interested in the complexity of algorithms related to the $\mu$-calculus. See \cite{kupke2021sizematters}.}
The evaluation game $\mathcal{G}(M,w\models \varphi)$ is a game for two players: Verifier and Refuter, denoted by $\mathsf{V}$ and $\mathsf{R}$ respectively.
The positions of the game are of the form $\tuple{\psi, v}$ with $\psi\in \mathrm{Sub}(\varphi)$ and $v\in W$.
The initial position is $\tuple{\varphi, w}$.
Each position $\tuple{\psi,v}$ is owned by a player, who makes the next move.
Table \ref{table::evaluation} summarizes the ownership of $\tuple{\psi,v}$ and admissible moves on it; both are determined by the construction of $\psi$.
On Table \ref{table::evaluation}, $\psi_X$ denotes the unique subformula of $\varphi$ such that $X$ occurs freely in $\psi_X$ and
$\eta X.\psi_X\in \mathrm{Sub}(\varphi)$.

Let $\rho$ be a run of an evaluation game $\mathcal{G}(M,w\models \varphi)$.
If $\rho$ is finite, $\mathsf{V}$ wins $\rho$ iff $\mathsf{R}$ cannot make a move and $\mathsf{R}$ wins $\rho$ iff $\mathsf{V}$ cannot make a move.
If $\rho$ is infinite, let $\eta X.\psi\in \mathrm{Sub}(\varphi)$ be a formula such that: positions of the form $\tuple{\eta X.\psi, v}$ appear infinitely many often in $\rho$; and, for all formula $\theta$ such that positions $\tuple{\theta,v}$ appear infinitely often in $\rho$, $\theta\in\mathrm{Sub}(\eta X.\psi)$.
Then $\mathsf{V}$ wins $\rho$ iff $\eta$ is $\nu$ and $\mathsf{R}$ wins $\rho$ iff $\eta$ is $\mu$.
A strategy is a function indicating how a player should move.
A winning strategy for $\mathsf{V}$ is a strategy $\sigma$ for $\mathsf{V}$ such that $\mathsf{V}$ wins all runs where they follow $\sigma$.
We define winning strategies for $\mathsf{R}$ similarly.

Relational semantics and game semantics are equivalent:
\begin{proposition}
    \label{prop::kripke-and-game-semantics}
    Let $M = \tuple{W,\{R_i\}_{i\in\Lambda},V}$ be a Kripke model, $w\in W$ be a world, and $\varphi$ be a $\mu$-formula.
    Then $M, w\models\varphi$ iff $\mathsf{V}$ has a winning strategy in the evaluation game $\mathcal{G}(M,w\models\varphi)$; and $M,w\not\models\varphi$ iff $\mathsf{R}$ has a winning strategy in the evaluation game $\mathcal{G}(M,w\models\varphi)$.
\end{proposition}
\begin{proof}
    See \cite{bradfield2018mucalculus} or \cite{pacheco2023game}
\end{proof}

\begin{table}[htb]\renewcommand\arraystretch{1.2}
  \begin{center}

  \caption{The rules of evaluation game for modal $\mu$-calculus.}
  \label{table::evaluation}
  \begin{tabular}{c|c||c|c}
    \multicolumn{2}{c}{Verifier} & \multicolumn{2}{c}{Refuter}\\
    \hline
    Position &  Admissible moves & Position &  Admissible moves\\

    $\tuple{\psi_1 \lor \psi_2, w}$ &  $\{\tuple{ \psi_1, w} , \tuple{\psi_2, w} \}$ &
    $\tuple{\psi_1 \land \psi_2 , w}$ &  $\{ \tuple{\psi_1 , w} , \tuple{\psi_2 , w} \}$ \\

    $\tuple{\Diamond_i \psi , w}$ & $\{ \tuple{\psi, v} \mid \tuple{w, v} \in R_i \}$ &
    $\tuple{\Box_i \psi, w}$ & $\{ \tuple{\psi, v} \mid \tuple{w, v} \in R_i \}$ \\

    $\tuple{P, w}$ and $w \not\in V(P)$  &  $\emptyset$ &
    $\tuple{P, w}$ and $w \in V(P)$  &  $\emptyset$ \\

    $\tuple{\neg P, w}$ and $w \in V(P)$  &  $\emptyset$ &
    $\tuple{\neg P, w}$ and $w \not\in V(P)$  &  $\emptyset$ \\

    $\tuple{\mu X.\psi_X, w}$ &  $\{\tuple{\psi_X, w} \}$ & $\tuple{\nu X.\psi_X, w}$ &  $\{\tuple{\psi_X, w} \}$ \\

    $\tuple{X, w}$ & $\{ \tuple{\mu X.\psi_X, w} \}$ &
    $\tuple{X, w}$ & $\{ \tuple{\nu X.\psi_X, w} \}$ \\

  \end{tabular}
  \end{center}
\end{table}

\paragraph{Parity games}
A parity game is a tuple $\mathcal{P}=\tuple{V_\exists, V_\forall, v_0, E, \Omega}$ where two players $\exists$ and $\forall$ move a token in the graph $\tuple{V_\exists \cup V_\forall, E}$.
We suppose $V_\exists$ and $V_\forall$ are disjoint sets of \emph{vertices}; $E\subseteq (V_\exists\cup V_\forall)^2$ is a set of \emph{edges}; and $\Omega:V_\exists\cup V_\forall \to n$ is a \emph{parity function}.
If a player has no available move, then the other player wins.
In an infinite play $\rho$, the winner is determined by the following parity condition: $\exists$ wins $\rho$ iff the greatest parity which appears infinitely often in $\rho$ is even; otherwise, $\forall$ wins $\rho$.
$\exists$ wins the parity game $\mathcal{P}$ iff $\exists$ has a winning strategy; a winning strategy for $\exists$ is a function $\sigma$ from $V_\exists$ to $V_\exists\cup V_\forall$, where, if $\exists$ follows $\sigma$, all resulting plays are winning for them.
Similarly, $\forall$ wins $\mathcal{P}$ iff $\forall$ has a winning strategy.

Fix a parity game $\mathcal{P}=\tuple{V_\exists, V_\forall, v_0, E, \Omega}$.
The set of winning positions for $\exists$ in $\mathcal{P}$ is the set of positions $v$ where $\exists$ wins the parity game if the players start at $v$.
That is, $v\in V_\exists\cup V_\forall$ is a winning position for $\exists$ iff $\exists$ wins $\mathcal{P}_v=\tuple{V_\exists, V_\forall, v, E, \Omega}$.

Sometimes it is convenient to suppose that all parity games are \emph{tree-like}.
That is, for all $v\in V_\exists\cup V_\forall$, there is no path $v = v_0 E \cdots E v_n = v$, for all $n\in\mathbb{N}$.
Any parity game $\mathcal{P}=\tuple{V_\exists, V_\forall, v_0, E, \Omega}$ can be unfolded into a tree-like parity game.
In the unfolded game, instead of moving to a node $v$, the players move to a fresh copy of $v$.
The unfolded parity game is bisimilar to the original game.

\section{Evaluation games as parity games}
\label{sec::evaluation-games-as-parity-games}
Fix a model $M=\tuple{W,\{R_i\}_{i\in\Lambda},V}$, a world $w\in W$ and a $\mu$-formula $\varphi$.
We define a parity game $\mathcal{G}^\mathrm{P} = \mathcal{G}^\mathrm{P}(M,w\models\varphi)=\tuple{V_\exists, V_\forall, v_0, E, \Omega}$ which is equivalent to $\mathcal{G} = \mathcal{G}(M,w\models\varphi)$.

The set of positions $V_\exists$ consists of the positions owned by $\verifier$ in $\mathcal{G}$.
Similarly $V_\forall$ consists of the positions owned by $\forall$ in $\mathcal{G}$.
The set of edges $E$ consists of the transitions in $\mathcal{G}$.
The initial position $v_0$ is $\tuple{\varphi,w}$.
Define the parity function:
\begin{itemize}
    \item $\Omega(\tuple{\mu X.\psi,v}) = 2(i+\varepsilon)-1$ if $\mu X.\psi\in \Sigma^\mu_{2i+\varepsilon}\setminus\Pi^\mu_{2i+\varepsilon}$;
    \item $\Omega(\tuple{\nu X.\psi,v}) = 2i$ if $\nu X.\psi\in \Pi^\mu_{2i+\varepsilon}\setminus\Sigma^\mu_{2i+\varepsilon}$;
    \item $\Omega(\tuple{\psi,v}) = 0$ for $\psi$ not of the form $\eta X.\psi$;
\end{itemize}
where $\varepsilon \in \{0,1\}$.

\begin{proposition}
    \label{prop::eval-is-equiv-to-parity}
    Let $M=(W,\{R_i\}_{i\in\Lambda},V)$ be a Kripke model, $w\in W$, and $\varphi$ a $\mu$-formula.
    Then:
    \[
        \verifier \text{ wins } \mathcal{G}(M,w\models\varphi) \iff \exists \text{ wins } \mathcal{G}^\mathrm{P}(M,w\models\varphi).
    \]
\end{proposition}
\begin{proof}
    Denote $\mathcal{G}(M,w\models\varphi)$ by $\mathcal{G}$ and $\mathcal{G}^\mathrm{P}(M,w\models\varphi)$ by $\mathcal{G}^\mathrm{P}$.
    As both games are on the same board, strategies for $\verifier$ and $\refuter$ in $\mathcal{G}$ are strategies for $\exists$ and $\forall$ in $\mathcal{G}^\mathrm{P}$.
    As any position is owned by $\verifier$ in $\mathcal{G}$ iff it is owned by $\exists$ in $\mathcal{G}$, any finite run is winning for $\verifier$ iff it is winning for $\exists$.

    Consider an infinite run $\rho$.
    The parity $\Omega(\tuple{\psi,v})$ is odd iff $\psi\in\Sigma^k\setminus\Pi^k$ for some $k\in \mathbb{N}$.
    If the greatest infinitely often occurring parity in $\rho$ is odd, then some $\mu X.\psi$ is the outermost infinitely often occurring fixed-point formula.
    Otherwise, if $\mu X.\psi \in\mathrm{Sub(\nu Y.\theta)}$ and $\nu Y.\theta$ is the outermost infinitely occurring fixed-formula formula, then $\Omega(\tuple{\nu Y.\theta,v})\geq \Omega(\tuple{\mu X.\psi,v})$ and $\Omega(\tuple{\nu Y.\theta,v})$ is even.
    Similarly, if the greatest infinitely often occurring parity in $\rho$ is even, then some $\nu X.\psi$ is the outermost infinitely often occurring fixed-point formula.
    Either way, $\rho$ is winning for $\verifier$ in $\mathcal{G}$ iff $\rho$ is winning for $\exists$ in $\mathcal{G}^\mathrm{P}$.
\end{proof}

\section{Winning region formulas}
\label{sec::winning-region-formulas}
Let $\mathsf{F}_0$ and $\mathsf{F}_1$ be classes of frames with signatures $\{0\}$ and $\{1\}$, respectively.
Suppose that $\circ\leftarrow\circ\to\circ$ is a subframe of $\mathsf{F}_0$ and that $\circ\to\circ$ is a subframe of $\mathsf{F}_1$.
Fix $F_0\in \mathsf{F}_0$ and $F_1\in \mathsf{F}_1$ witnessing these facts.
Given a parity game $\mathcal{P}$ we will define an associated Kripke model $\mathcal{P}^\mathrm{K}$ with frame in $\mathsf{F}_0\otimes\mathsf{F}_1$.
We will also define winning region $\mu$-formulas $W_n$, for all $n\in\mathbb{N}$.
If $\mathcal{P}$ is a parity game which uses parities up to $n$, then $\exists$ wins $\mathcal{P}$ starting at $v$ iff $\mathcal{P}^\mathrm{K},v \models W_n$.

Let $\mathcal{P}=\tuple{V_\exists, V_\forall, v, E, \Omega}$ be a parity game.
We represent $\mathcal{P}$ as a birelational Kripke model $\mathcal{P}^K =\tuple{W,R_0,R_1,V}$.
The set $W$ of possible worlds will consist of a world $\underline{v}$ for each state $v\in V_\exists, V_\forall$ and a countable supply of other worlds.
If $v\in V_\exists, V_\forall$ and $vE = \{v_0, \dots v_n\}$, then we will represent the connection between $v$ and the $v_i$ using fresh isomorphic copies of $F_0$ and $F_1$.
We first use a copy of $F_0$ to choose between $v_0$ and the other vertices, then we use copies of $F_1$ to confirm the choices.
Similarly, we use a copy of $F_0$ to choose between $v_1$ and the other vertices, and copies of $F_1$ to confirm the choices.
We repeat this procedure until we use up all the $v_i$.
By using fresh copies of $F_0$ and $F_1$, we guarantee that the resulting frame is in $\mathsf{F}_0 \otimes \mathsf{F}_1$.
We denote by $\underline{v}, \underline{v}_0, \cdots$ the worlds of $\mathcal{P}^\mathrm{K}$ corresponding to the positions $v, v_0, \cdots$; we do not name the other worlds connecting them.
An example of this construction is depicted in Figure \ref{figure::bimodal-strictness-bifurcation}.

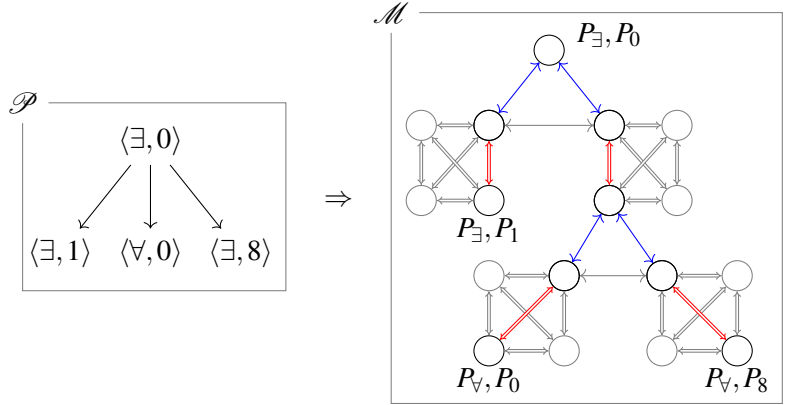
\begin{figure}
\centering
    \usetikzlibrary {arrows.meta}
    \tikzstyle{world}=[circle,draw,minimum size=4mm,inner sep=0pt]
    \tikzstyle{first}=[<->,blue]
    \tikzstyle{second}=[double, arrows = {Implies-Implies},red]
    \tikzstyle{ignore}=[gray]
      \begin{tikzpicture}
        \node (m) at (-0.5,0.5) {$\mathcal{P}$};

        \node (w) at (1.2,  0) {$\tuple{\exists,0}$};
        \node (v) at (0, -1.5) {$\tuple{\exists,1}$};
        \node (u) at (1.2, -1.5) {$\tuple{\forall,0}$};
        \node (r) at (2.4, -1.5) {$\tuple{\exists,8}$};

        \draw[->] (w) -- (v);
        \draw[->] (w) -- (u);
        \draw[->] (w) -- (r);

        \draw[-,color=gray] (m) -- (-0.5,-2) -- (3,-2) -- (3, 0.5) -- (m);

        \node (arrow) at (3.7, -0.8) {$\Rightarrow$};

      \begin{scope}[shift={(5,1.2)}]
        \node (m) at (-0.6,0.5) {$\mathcal{M}$};

        \node[world] (w1) at (1.5,  0) {};
        \node[world] (w2) at (0.7, -1) {};
        \node[world] (w3) at (2.3, -1) {};

        \draw[first] (w1) -- (w2);
        \draw[first] (w1) -- (w3);
        \draw[first,ignore] (w2) -- (w3);

        \node[world] (v1) at (0.7, -1) {};
        \node[world] (v2) at (0.7, -2) {};
        \node[world,ignore] (v3) at (-0.2  , -1) {};
        \node[world,ignore] (v4) at (-0.2  , -2) {};

        \draw[second] (v1) -- (v2);
        \draw[second,ignore] (v1) -- (v3);
        \draw[second,ignore] (v1) -- (v4);
        \draw[second,ignore] (v2) -- (v3);
        \draw[second,ignore] (v2) -- (v4);
        \draw[second,ignore] (v3) -- (v4);

        \node[world] (u1) at (2.3, -1) {};
        \node[world] (u2) at (2.3, -2) {};
        \node[world,ignore] (u3) at (3.2  , -1) {};
        \node[world,ignore] (u4) at (3.2  , -2) {};

        \draw[second] (u1) -- (u2);
        \draw[second,ignore] (u1) -- (u3);
        \draw[second,ignore] (u1) -- (u4);
        \draw[second,ignore] (u2) -- (u3);
        \draw[second,ignore] (u2) -- (u4);
        \draw[second,ignore] (u3) -- (u4);

        \node[world] (r1) at (2.3, -2) {};
        \node[world] (r2) at (1.7,   -3) {};
        \node[world] (r3) at (3, -3) {};

        \draw[first] (r1) -- (r2);
        \draw[first] (r1) -- (r3);
        \draw[first,ignore] (r2) -- (r3);

        \node[world] (s1) at (1.7, -3) {};
        \node[world,ignore] (s2) at (1.7, -4) {};
        \node[world,ignore] (s3) at (0.7  , -3) {};
        \node[world] (s4) at (0.7  , -4) {};

        \draw[second,ignore] (s1) -- (s2);
        \draw[second,ignore] (s1) -- (s3);
        \draw[second,ignore] (s2) -- (s3);
        \draw[second,ignore] (s2) -- (s4);
        \draw[second,ignore] (s3) -- (s4);
        \draw[second] (s1) -- (s4);

        \node[world] (s1) at (3, -3) {};
        \node[world,ignore] (s2) at (3, -4) {};
        \node[world,ignore] (s3) at (4  , -3) {};
        \node[world] (s4) at (4  , -4) {};

        \draw[second,ignore] (s1) -- (s2);
        \draw[second,ignore] (s1) -- (s3);
        \draw[second,ignore] (s2) -- (s3);
        \draw[second,ignore] (s2) -- (s4);
        \draw[second,ignore] (s3) -- (s4);
        \draw[second] (s1) -- (s4);

        \node (w-val) at (2.3,  0.2) {$P_\exists,P_0$};
        \node (v-val) at (0.7,  -2.4) {$P_\exists,P_1$};
        \node (u-val) at (0.7,  -4.4) {$P_\forall,P_0$};
        \node (r-val) at (4,  -4.4) {$P_\forall,P_8$};

        \draw[-,color=gray] (m) -- (-0.6,-4.7) -- (4.6,-4.7) -- (4.6, 0.5) -- (m);
      \end{scope}
      \end{tikzpicture}
    \caption{Example of a parity game $\mathcal{P}$ and the corresponding bimodal model $\mathcal{P}^K$.
    The model $\mathcal{P}^K$ is built using copies of $\mathsf{S5}$ models.}
    \label{figure::bimodal-strictness-bifurcation}
\end{figure}

We will use fresh propositional symbols $\bd$, $\st$, $\prezero$, $\preone$, $\nextzero$, and $\nextone$ when defining $\mathcal{P}^K$.
The proposition symbol $\bd$ indicate that a world is used to represent the parity game.
That is, only the isomorphic copies of $\circ\leftarrow\circ\to\circ$ and $\circ\to\circ$ used in the paragraph above satisfy $\bd$.
The proposition symbol $\st$ indicates that a world corresponds to a position in the parity game.
That is, it holds only on worlds which are $\underline{v}$ for some $v\in V_\exists\cup V_\forall$.
The proposition symbols $\prezero$, $\preone$, $\nextzero$, and $\nextone$ are used to represent the direction of the moves in the parity game in the Kripke model.
$\prezero$ holds when we are making a choice and $\nextzero$ holds after we make a choice.
Similarly, $\preone$ holds when we are confirming a choice and $\nextone$ holds after we confirmed a choice.
These propositional symbols allow us to stay in the part of the model which represents the parity game.
They will also guarantee that sequences moves in evaluation games over $\mathcal{P}^K$ correspond to moves in $\mathcal{P}$.

The proposition symbols $P_\exists$ and $P_\forall$ indicate the ownership of the positions: $P_\exists$ holds at $\underline{v}$ iff $v\in V_\exists$ and  $P_\forall$ holds at $\underline{v}$ iff $v\in V_\forall$.
The proposition symbols $P_0, \dots, P_n$ indicate the parities of the positions: $P_i$ holds at $\underline{v}$ iff $\Omega(v) = i$.
At each $\underline{v}$, exactly one of the $P_i$ will hold.
The proposition symbols $P_\exists$, $P_\forall$, $P_0,\dots, P_n$ are false at worlds which are not of the form $\underline{v}$ for some $v\in V_\exists \cup V_\forall$.
This finishes the definition of $\mathcal{P}^\mathrm{K}$.

To define the winning region formulas $W_n$, we use the following shorthand formulas:
\begin{itemize}
    \item $\blacklozenge \varphi := \nu Y.\prezero \land \bd \land \Diamond_0(\nextzero \land \preone \land \bd\land \Diamond_1 (\nextone \land \bd \land ((Y \land \neg\st) \lor (\varphi \land \st))))$; and
    \item $\blacksquare \varphi := \nu Y.\prezero \land \bd \to \Box_0(\nextzero \land \preone \land \bd\to \Box_1 (\nextone \land \bd \to ((Y\land \neg\st) \land (\varphi\land \st))))$,
\end{itemize}
where $Y$ is a fresh variable symbol.
We use these modalities to represent a move in $\mathcal{P}$ as multiple moves in evaluation game $\mathcal{P}^\mathrm{K},\underline{v}\models W_n$.
Given $n\in \mathbb{N}$, define:
\begin{align*}
    W_n := &\eta X_n\dots \nu X_2 \mu X_1 \nu X_0.\bigvee_{0\leq j\leq n}[(P_j \land P_{\exists} \land \blacklozenge X_j) \lor (P_j \land P_{\forall} \land \blacksquare X_j)].
\end{align*}
The formula $W_n$ defines the winning positions of $\exists$ in parity games using parities up to $n$:
\begin{proposition}
    \label{prop::win-reg-fmls-are-correct}
    Let $\mathcal{P}=\tuple{V_\exists, V_\forall, v_0, E, \Omega}$ be a parity game.
    If $\max\{\Omega(v) \mid v\in W\}\leq n$, then
    \[
        \mathcal{P}^K,\underline{v}_0\models W_n \text{ iff } \text{$\exists$ wins $\mathcal{P}$.}
    \]
\end{proposition}

\begin{proof}
    Suppose $\mathcal{P}^\mathrm{K},\underline{v}_0\models W_n$.
    Let $\sigma$ be a winning strategy for $\verifier$ in the evaluation game $\mathcal{G} := \mathcal{G}(\mathcal{P}^\mathrm{K},\underline{v}_0\models W_n)$.
    We define a winning strategy $\sigma'$ for $\exists$ in $\mathcal{P}$ while playing simultaneous runs of $\mathcal{G}$ and $\mathcal{P}$.

    The games $\mathcal{G}$ and $\mathcal{P}$ start at positions $\tuple{W_n,\underline{v}_0}$ and $v_0$, respectively.
    First, have the players move to the position
    \[
        \left\langle \bigvee_{0\leq j\leq n}[(P_j \land P_{\exists} \land \blacklozenge X_j) \lor (P_j \land P_{\forall} \land \blacksquare X_j), \underline{v}_0\right\rangle
    \]
    in $\mathcal{G}$.

    Now, suppose the players are at positions
    \[
        \left\langle \bigvee_{0\leq j\leq n}[(P_j \land P_{\exists} \land \blacklozenge X_j) \lor (P_j \land P_{\forall} \land \blacksquare X_j), \underline{v}\right\rangle
    \]
    in $\mathcal{G}$ and $v$ in $\mathcal{P}$, respectively.
    As $\sigma$ is winning for $\verifier$ in $\mathcal{G}$, $\sigma$ does not make any immediately losing move.
    That is, $\verifier$ picks the disjuncts according to $v$'s parity and owner.
    We also have $\forall$ make non-immediately losing moves.
    The players eventually reach one of two possible cases:

    \emph{Case 1.} The players are in the position $\tuple{\blacklozenge X_j, \underline{v}}$ in $\mathcal{G}$, with $v\in V_\exists$.
    By our choice of $\sigma$, $\verifier$ must eventually reach a position of the form $\tuple{X_j,\underline{v'}}$.
    Then $\sigma'$ tells $\exists$ to move to $v'$ in $\mathcal{P}$.

    \emph{Case 2.} The players are in the position $\tuple{\blacksquare X_j,\underline{v}}$ in $\mathcal{G}$ and $v\in V_\forall$ in $\mathcal{P}$.
    If $\forall$ moves to $v'$, $\refuter$ moves to $\tuple{X_j,\underline{v'}}$ in $\mathcal{G}$ in finitely many steps.

    Now, have the players regenerate $X_j$ in $\mathcal{G}$ and move until they get to positions of the form
    \[
        \left\langle \bigvee_{0\leq j\leq n}[(P_j \land P_{\exists} \land \blacklozenge X_j) \lor (P_j \land P_{\forall} \land \blacksquare X_j), \underline{v'}\right\rangle\text{ and } v'
    \]
    in $\mathcal{G}$ and $\mathcal{P}$, respectively.
    We are back to the initial situation, and we repeat this process to define $\sigma'$.

    We consider parallel runs $\rho$ in $\mathcal{G}$ and $\rho'$ in $\mathcal{P}$ played according to $\sigma$ and $\sigma'$, respectively.
    Then either both runs are finite or both runs are infinite.
    If $\rho'$ is finite, this means that one of the players didn't have a move available to play at a position $v$ in $\mathcal{P}$.
    Therefore, one of the players couldn't find a valid position to play after $\tuple{\blacklozenge X_{\Omega(v)}, \underline{v}}$ or $\tuple{\blacksquare X_{\Omega(v)}, \underline{v}}$.
    The former is not possible by our choice of $\sigma$, so it must be $\forall$ who could not make a move.
    Therefore $\exists$ wins $\rho'$.
    If $\rho$ is infinite, then the outermost infinitely often regenerated fixed-point operator is some $\nu X_{2k}$.
    By the construction of $\sigma'$ the greatest infinitely often occurring parity must be $2k$.
    Therefore $\exists$ wins $\rho'$.
    We can now conclude that $\sigma'$ is a winning strategy for $\exists$ in $\mathcal{P}$.

    On the other hand, suppose $\exists$ wins $\mathcal{P}$ via $\sigma'$.
    We define $\sigma$ for $\verifier$ in $\mathcal{G}$.
    At vertices of the form $\tuple{\blacklozenge X_j,\underline{v}}$ in $\mathcal{G}$, have $\verifier$ move to
    \[
        \sigma(\tuple{\blacklozenge X_j,\underline{v}}) := \tuple{X_j,\underline{v'}},
    \]
    with $v'= \sigma'(v)$.
    On other positions, have $\sigma$ be the non-immediately losing moves for $\verifier$.

    Consider parallel runs $\rho$ in $\mathcal{G}$ and $\rho'$ in $\mathcal{P}$ played according to $\sigma$ and $\sigma'$, respectively.
    If $\rho$ is finite, then one of the players made a move which invalidates one of the auxiliary propositions, or did not have an adequate moves after a position of the form $\tuple{\blacklozenge X_j,\underline{v}}$ or $\tuple{\blacksquare X_j,\underline{v}}$.
    By the choice of $\sigma'$ and definition of $\sigma$, $\verifier$ makes no such move.
    So it must be $\refuter$ who made such move and lost the game; therefore $\verifier$ wins.
    If $\rho$ is infinite, the greatest parity appearing infinitely often in $\rho'$ is even.
    Therefore the outermost infinitely often occurring fixed-point operator in $\rho$ is a $\nu$-operator.
    $\rho$ is winning for $\verifier$.
    Therefore $\sigma$ is a winning strategy for $\verifier$ in $\mathcal{G}$.
\end{proof}

Given an evaluation game $\mathcal{G}(M,w\models\varphi)$, we define the Kripke model $\mathcal{G}^\mathrm{K}(M,w\models\varphi)$ as $(\mathcal{G}^\mathrm{P}(M,w\models\varphi))^\mathrm{K}$.
As evaluation games are also parity games, the $W_n$ also define winning regions for $\mathsf{V}$ in evaluation games:
\begin{proposition}
    \label{prop::win-reg-fmls-are-correct-for-eval-games}
    Let $M=(W,R_0,R_1,V)$ be a bimodal Kripke model, $w\in W$, and $\varphi$ a bimodal $\mu$-formula.
    If $n\geq 1$ and the greatest parity used in $\mathcal{G}^\mathrm{P}(M,w\models\varphi)$ is less or equal than $n$, then:
    \[
        M,w\models \varphi \text{ iff } \mathcal{G}^K(M, w\models \varphi), \tuple{\varphi,w}\models W_n.
    \]
\end{proposition}
\begin{proof}
    We have:
    \begin{align*}
        M,w\models \varphi
        &\text{ iff } \verifier\text{ wins }\mathcal{G}(M,w\models \varphi) \\
        &\text{ iff } \exists\text{ wins }\mathcal{G}^\mathrm{P}(M,w\models \varphi) \\
        &\text{ iff } \mathcal{G}^\mathrm{K}(M,w\models \varphi), \tuple{\varphi,w}\models W_n.
    \end{align*}
    The first equivalence follows from Proposition \ref{prop::kripke-and-game-semantics}, the second one follows from Proposition \ref{prop::eval-is-equiv-to-parity}, the third one follows from Proposition \ref{prop::win-reg-fmls-are-correct}.
\end{proof}

\section{Strictness}
\label{sec::strictness}
Fix classes of frames $\mathsf{F}_0$ and $\mathsf{F}_1$ with signatures $\{0\}$ and $\{1\}$, respectively.
We show that, if $\circ\leftarrow\circ\to\circ$ is a subframe of $\mathsf{F}_0$ and that $\circ\to\circ$ is a subframe of $\mathsf{F}_1$, then the $\mu$-calculus' alternation hierarchy is strict over $\mathsf{F}_0\otimes\mathsf{F}_1$.

Let $(M,w) = \tuple{W,R_0,R_1,V,w}$ and $(M',w') = \tuple{W',R_0',R_1',V',w'}$ be pointed Kripke models without loops in their graphs.
$(M,w)$ is \emph{isomorphic} to $(M',w')$ iff there is a bijection $I:W\to W'$ such that:
\begin{itemize}
  \item $I(w)=w'$;
  \item for all $v,v'\in W$, $v R_0 v'$ iff $I(v) R_0' I(v')$;
  \item for all $v,v'\in W$, $v R_1 v'$ iff $I(v) R_1' I(v')$; and
  \item for all $v\in W$, $v\in V(P)$ iff $I(v) \in V'(P)$.
\end{itemize}
For all $n\in\mathbb{N}$, let $(M\upharpoonright n,w)$ be the submodel of $(M,w)$ obtained by restricting $W$ to worlds with distance less than $n$ from $w$.
We say $(M,w)$ is \emph{$n$-isomorphic} to $(M',w')$ if and only if $(M\upharpoonright n,w)$ is isomorphic to $(M'\upharpoonright n,w')$.
For any $(M,w)$, $(M\restrictedto 0, w)$ is an empty Kripke model.
We assume the empty Kripke model is isomorphic to itself.

Given a $\mu$-formula $\varphi$, let $f_\varphi$ be the function mapping a pointed model to the pointed Kripke model representing its evaluation game with respect to $\varphi$.
That is $f_{\varphi}(M,w) = (\mathcal{G}^\mathrm{K}(M,w\models\varphi),\tuple{\varphi,w})$, for all pointed models $(M,w)$.
\begin{lemma}
    \label{lem::n-iso-then-f-is-n+1-iso}
    Fix a $\mu$-formula $\varphi$.
    If $(M,w)$ and $(M',w')$ are $n$-isomorphic via a function $I$, then $f_{\varphi\land\varphi}(M,w)$ and $f_{\varphi\land\varphi}(M',w')$ are $(n+1)$-isomorphic via the function $J$ defined by:
    \[
      J(\tuple{\psi,w}) = (\tuple{\psi,I(w)}),
    \]
    for all world $w$ of $M$ and subformula $\psi$ of $\varphi$.
\end{lemma}
\begin{proof}
    As $(M,w)$ and $(N,v)$ are $n$-isomorphic, the evaluation games $\mathcal{G}(M,w\models\varphi\land\varphi)$ and $\mathcal{G}(N,v\models\varphi\land\varphi)$ are going to be same up to $n$-many plays of the form $\tuple{\triangle\psi,w}$, with $\triangle\in\{\Box_0,\Diamond_0,\Box_1,\Diamond_1\}$.
    As the first move in an evaluation game for the formula $\varphi\land\varphi$ is to choose between a conjunction, we can guarantee that the two games above are the same up to $n+1$ moves.
\end{proof}

\begin{lemma}
    \label{lem::fixpoint-exists}
    For all $\mu$-formula $\varphi$, the function $f_{\varphi\land\varphi}$ has a fixed-point (up to isomorphism).
    That is, there is a model $(M,w)$ such that $f_\varphi(M,w)$ is isomorphic to $(M,w)$.
\end{lemma}
\begin{proof}
    Let $(M_0,w_0)$ be a fixed arbitrary pointed Kripke model.
    We define $(M_{n+1},w_{n+1}) = f_{\varphi\land\varphi}(M_n,w_n)$ inductively on $n\in\mathbb{N}$.
    If $n=0$, then $(M_0,w_0)$ and $(M_1,w_1)$ are trivially $0$-isomorphic.
    By induction on $n$, $(M_n,w_n)$ and $(M_{n+1},w_{n+1})$ are $n$-isomorphic via Lemma \ref{lem::n-iso-then-f-is-n+1-iso}.
    Therefore, if $m> n$ then $(M_n,w_n)$ is $n$-isomorphic to $(M_m,w_m)$.

    We can now define a pointed Kripke model $(M,w)$ which is $n$-isomorphic to $(M_n,w_n)$ for all $n$.
    We identify $(M_n\upharpoonright n,w_n)$ and $(M_{n+1}\upharpoonright n,w_{n+1})$, since they are $n$-isomorphic by the restriction of the isomorphism $J_n$ from Lemma \ref{lem::n-iso-then-f-is-n+1-iso}.
    Furthermore, the isomorphisms $J_n$ and $J_{n+1}$ coincide on $(M_n\upharpoonright n,w_n)$ by construction.
    Let $M$ be the Kripke model whose graph is the union of the graphs of the models $M_{n}\upharpoonright n$ and whose valuation is the union of the valuation of the $M_{n}\upharpoonright n$; also set $w=w_0$.
    Then $f_{\varphi\land\varphi}(M,w)$ is the Kripke model whose graph is the union of the graphs of the models $M_{n+1}\upharpoonright n+1$ and whose valuation the union of the valuation of the $M_{n+1}\upharpoonright n+1$.
    The union of the $J_n$ is an isomorphism between $(M,w)$ and $f_{\varphi\land\varphi}(M,w)$.
\end{proof}

\begin{proof}[Proof of Item 1 of the Main Theorem]
    Let $\mathsf{F}_0$ and $\mathsf{F}_1$ be classes of unimodal Kripke frames closed under isomorphic copies and disjoint unions.
    Suppose $\circ\leftarrow\circ\to\circ$ is a subframe of $\mathsf{F}_0$ and $\circ\to\circ$ a subframe of $\mathsf{F}_1$.

    If $n$ is even, then $W_n \in {\Pi}^{\mu}_{n+1}$.
    For a contradiction, suppose that $W_n$ is equivalent to some formula in ${\Pi}^{\mu}_{n}$ over $\mathsf{F}_0\otimes\mathsf{F}_1$.
    Let $\varphi\in{\Sigma}^{\mu}_{n}$ be equivalent to $\neg W_n$.
    Let $(M,w)$ be a fixed-point of $f_{\varphi\land\varphi}$.
    Then
        \begin{align*}
            M,w\models\neg W_n
            & \iff M,w\models \varphi\land\varphi \\
            & \iff f_{\varphi\land\varphi}(M,w)\models W_n \\
            & \iff M,w\models W_n.
        \end{align*}
    The second equivalence follows from Proposition \ref{prop::win-reg-fmls-are-correct} and the third one follows from Lemma \ref{lem::fixpoint-exists}.
    This is a contradiction, and so $W_n$ is not equivalent to any formula in ${\Pi}^{\mu}_{n}$ over $\mathsf{F}_0\otimes\mathsf{F}_1$.
    The case for $n$ odd is symmetric: $W_n \in {\Sigma}^{\mu}_{n+1}$ and is not equivalent to any formula in ${\Sigma}^{\mu}_{n}$.
\end{proof}

\section{Finishing the proof of the Main Theorem}
\label{sec::other-items}
To prove Items 2 and 3 of the Main Theorem, we modify two points in the proof above:
first, we define new functions transforming parity games into Kripke models; second, we supply new versions of the modalities $\blacklozenge$ and $\blacksquare$.

We first consider the case of Item 2.
Let $\mathsf{F}_0$ and $\mathsf{F}_1$ be classes of unimodal Kripke frames closed under isomorphic copies and disjoint unions.
Suppose $\circ\to\circ\to\circ$ is a subframe of $\mathsf{F}_0$ and $\circ\to\circ$ is a subframe of $\mathsf{F}_1$.
When we define a Kripke model $\mathcal{P}^\mathrm{K}$ from a parity game $\mathcal{P}$, we change the way we use the copies of $\circ\to\circ\to\circ$ and $\circ\to\circ$.
Suppose $v\in V_\exists, V_\forall$ and $vE = \{v_0, \dots v_n\}$.
The players choose the next position as follows: they first move once in a copy of $\circ\to\circ\to\circ$; they then confirm some $v_i$ using a copy of $\circ\to\circ$ or move along the current copy $\circ\to\circ\to\circ$; if they moved along $\circ\to\circ\to\circ$, they must confirm this move via a copy of $\circ\to\circ$.

To control the movement of the players over copies of $\circ\to\circ\to\circ$, we use three propositional symbols $\prezero$, $\midzero$, and $\nextzero$.
Here, $\prezero$ holds at the first world of the copies of $\circ\to\circ\to\circ$, $\midzero$ holds at the second world, and $\nextzero$ holds at the third world.
We define $\blacklozenge$ and $\blacksquare$ as follows:
\begin{itemize}
    \item $\blacklozenge \varphi := \mu Y.\prezero \land \bd \land \Diamond_0[\midzero \land \preone \land  \land \bd \land \Diamond_1(\nextone \land \bd \land ((Y\land \neg\st) \lor (\varphi \land \st))) \lor \Diamond_0(\nextzero \land \preone \land \bd \land \Diamond_1(\nextone \land \bd \land ((Y\land \neg\st) \lor (\varphi \land \st))))]$; and
    \item $\blacksquare \varphi := \mu Y.\prezero \land \bd \to \Box_0[\midzero \land \preone \land  \land \bd \to \Box_1(\nextone \land \bd \land ((Y\land \neg\st) \lor (\varphi \land \st))) \land \Box_0(\nextzero \land \preone \land \bd \to \Box_1(\nextone \land \bd \land ((Y\land \neg\st) \lor (\varphi \land \st))))]$,
\end{itemize}
where $Y$ is a fresh variable symbol.
The definition of the winning region formulas $W_n$ are the same as above, where $\blacklozenge$ and $\blacksquare$ use their new definitions.

Now for the proof of Item 3.
Let $\mathsf{F}_0$, $\mathsf{F}_1$, and $\mathsf{F}_2$ be classes of unimodal Kripke frames closed under isomorphic copies and disjoint unions.
Suppose $\circ\to\circ$ is a subframe of $\mathsf{F}_0$, $\mathsf{F}_1$, and $\mathsf{F}_2$.
Given $v\in V_\exists, V_\forall$ and $vE = \{v_0, \dots v_n\}$, we build a Kripke model as in the proof of Item 2, but instead of using a copy of $\circ\to\circ\to\circ$, we use two copies of $\circ\to\circ$, one from $\mathsf{F}_0$ and one from $\mathsf{F}_1$; we use copies of $\circ\to\circ$ from $\mathsf{F}_2$ to confirm the choices.
This time we will also use fresh proposition variables $\pretwo$ and $\nexttwo$ to control the movement of the players along copies of $\circ\to\circ$ in $\mathsf{F}_2$.
Here, we define $\blacklozenge$ and $\blacksquare$ as follows:
\begin{itemize}
    \item $\blacklozenge \varphi := \mu Y.\prezero \land \bd \land \Diamond_0[\nextzero \land \preone \land \pretwo \land \bd \land \Diamond_2(\nexttwo \land \bd \land ((Y\land \neg\st) \lor (\varphi \land \st))) \lor \Diamond_1(\nextone \land \pretwo \land \bd \land \Diamond_2(\nexttwo \land \bd \land ((Y\land \neg\st) \lor (\varphi \land \st))))]$; and
    \item $\blacksquare \varphi := \mu Y.\prezero \land \bd \to \Box_0[\nextzero \land \preone \land \pretwo \bd \to \Box_2(\nexttwo \land \bd \land ((Y\land \neg\st) \lor (\varphi \land \st))) \land \Box_1(\nextone \land \pretwo \land \bd \to \Box_2(\nexttwo \land \bd \land ((Y\land \neg\st) \lor (\varphi \land \st))))]$,
\end{itemize}
where $Y$ is a fresh variable symbol.
The definition of the winning region formulas $W_n$ are the same as above, where $\blacklozenge$ and $\blacksquare$ use their new definitions.

\section{Case studies on the collapse over multimodal logics}
\label{sec::case-studies}
We now comment on two logics where the $\mu$-calculus collapses to modal logic.
These are not originally framed in the context of multimodal $\mu$-calculus.

\paragraph{Provability Logic} $\mathsf{GLP}$ is a multimodal provability logic with signature $\mathbb{N}$, first defined by Japaridze.
One of the possible arithmetical interpretations for each $\Box_n$ is as a provability predicate for $\mathsf{I}\Sigma_n$.
Each modality $\Box_n$ satisfies the necessitation rule and the axioms for the provability $\mathsf{GL}$: $\Box (P\to Q)\to (\Box P\to \Box Q)$ and $\Box(\Box P\to P)\to \Box P$.
While $\mathsf{GLP}$ contains the fusion of infinitely many copies of $\mathsf{GL}$, it is not a fusion logic: it also includes the axioms $\Box_m P\to \Box_n\Box_m P$, $\Diamond_m P\to \Box_n\Diamond_m P$, and $\Box_m P \to \Box_n P$, for all $m\leq n$.

Ignatiev \cite{ignatiev1993strong} proved that $\mathsf{GLP}$ has the fixed-point property: if $X$ is in the scope of some $\Box_i$ in $\varphi(X)$, then there is $\psi$ such that $\mathsf{GLP} \vdash \psi\leftrightarrow\varphi(\psi)$.
This implies that we do not get a more expressive logic if we add to it the operators $\mu$ and $\nu$.
While the additional conditions on the relation between the modalities makes it possible to have the fixed-point property, $\mathsf{GLP}$ is not complete over any class of Kripke models.

\paragraph{Intuitionistic Modal Logic}
$\mathsf{IS5}$ is an intuitionistic variation of $\mathsf{S5}$; it is also known as $\mathsf{MIPQ}$.
It consists of closure under necessitation and \emph{modus ponens} of the set of formulas containing the intuitionistic tautologies along with the axioms $T:=\Box\varphi\to \varphi \land \varphi\to \Diamond\varphi$, $4 := \Box\varphi\to \Box\Box\varphi \land \Diamond\Diamond\varphi\to \Diamond\varphi$, and $5 := \Diamond\varphi\to \Box\Diamond\varphi \land \Diamond\Box\varphi\to \Box\varphi$.
An $\mathsf{IS5}$ model is a tuple $\tuple{W,\preceq,R,V}$ satisfying: $\preceq$ is a pre-order; $R$ is an equivalence relation; $w R;\preceq v$ implies $w \preceq;R v$; and $w\preceq v$ and $w\in V(P)$ implies $v\in V(P)$.
$\mathsf{IS5}$ can be thought as a bimodal logic, where $\Box$ and $\Diamond$ are abbreviations for $\Box_\preceq\Box_R$ and $\Box_\preceq\Diamond_R$, respectively.
Ono \cite{ono1977intuitionistic} and Fischer Servi \cite{fischerservi1978finite} proved that $\mathsf{IS5}$ is complete over $\mathsf{IS5}$ frames.

Pacheco \cite{pacheco2023game} proved that the $\mu$-calculus collapses to constructive modal logic over $\mathsf{IS5}$ frames using game semantics for the constructive $\mu$-calculus.
This example shows that, if we add restrictions on how we use multiple modalities, then we may still have the collapse to modal logic.
Note that the relation $\tuple{W,\preceq}$ is an $\mathsf{S4}$ frame, and the $\mu$-calculus does not collapse to modal logic over $\mathsf{S4}$ frames \cite{alberucci2009modal}.
So the restriction on the usage of the modalities here is quite strong.

Motivated by the examples above, we close the paper with a problem:
\begin{problem}
    When does the $\mu$-calculus collapse to modal logic over multimodal frames?
\end{problem}


\bibliographystyle{eptcs}
\bibliography{multimodal-strictness}

\begin{thebibliography}{10}
\providecommand{\bibitemdeclare}[2]{}
\providecommand{\surnamestart}{}
\providecommand{\surnameend}{}
\providecommand{\urlprefix}{Available at }
\providecommand{\url}[1]{\texttt{#1}}
\providecommand{\href}[2]{\texttt{#2}}
\providecommand{\urlalt}[2]{\href{#1}{#2}}
\providecommand{\doi}[1]{doi:\urlalt{https://doi.org/#1}{#1}}
\providecommand{\eprint}[1]{arXiv:\urlalt{https://arxiv.org/abs/#1}{#1}}
\providecommand{\bibinfo}[2]{#2}

\bibitemdeclare{incollection}{alberucci2002strictness}
\bibitem{alberucci2002strictness}
\bibinfo{author}{Luca \surnamestart Alberucci\surnameend}
  (\bibinfo{year}{2002}): \emph{\bibinfo{title}{Strictness of the {{Modal}}
  {$\mu$}-{{Calculus Hierarchy}}}}.
\newblock In \bibinfo{editor}{Erich \surnamestart Gr{\"a}del\surnameend},
  \bibinfo{editor}{Wolfgang \surnamestart Thomas\surnameend} \&
  \bibinfo{editor}{Thomas \surnamestart Wilke\surnameend}, editors: {\slshape
  \bibinfo{booktitle}{Automata {{Logics}}, and {{Infinite Games}}: {{A Guide}}
  to {{Current Research}}}}, \bibinfo{series}{Lecture {{Notes}} in {{Computer
  Science}}}, \bibinfo{publisher}{{Springer}}, \bibinfo{address}{{Berlin,
  Heidelberg}}, pp. \bibinfo{pages}{185--201}, \doi{10.1007/3-540-36387-4_11}.

\bibitemdeclare{article}{alberucci2009modal}
\bibitem{alberucci2009modal}
\bibinfo{author}{Luca \surnamestart Alberucci\surnameend} \&
  \bibinfo{author}{Alessandro \surnamestart Facchini\surnameend}
  (\bibinfo{year}{2009}): \emph{\bibinfo{title}{The Modal {$\mu$}-Calculus
  Hierarchy over Restricted Classes of Transition Systems}}.
\newblock {\slshape \bibinfo{journal}{The Journal of Symbolic Logic}}
  \bibinfo{volume}{74}(\bibinfo{number}{4}), pp. \bibinfo{pages}{1367--1400},
  \doi{10.2178/jsl/1254748696}.

\bibitemdeclare{incollection}{bradfield1998simplifying}
\bibitem{bradfield1998simplifying}
\bibinfo{author}{Julian~C. \surnamestart Bradfield\surnameend}
  (\bibinfo{year}{1998}): \emph{\bibinfo{title}{Simplifying the Modal
  Mu-Calculus Alternation Hierarchy}}.
\newblock In \bibinfo{editor}{G.~\surnamestart Goos\surnameend},
  \bibinfo{editor}{J.~\surnamestart Hartmanis\surnameend},
  \bibinfo{editor}{J.~\surnamestart {van Leeuwen}\surnameend},
  \bibinfo{editor}{Michel \surnamestart Morvan\surnameend},
  \bibinfo{editor}{Christoph \surnamestart Meinel\surnameend} \&
  \bibinfo{editor}{Daniel \surnamestart Krob\surnameend}, editors: {\slshape
  \bibinfo{booktitle}{{{STACS}} 98}}, \bibinfo{volume}{1373},
  \bibinfo{publisher}{{Springer Berlin Heidelberg}}, \bibinfo{address}{{Berlin,
  Heidelberg}}, pp. \bibinfo{pages}{39--49}, \doi{10.1007/BFb0028547}.

\bibitemdeclare{incollection}{bradfield2018mucalculus}
\bibitem{bradfield2018mucalculus}
\bibinfo{author}{Julian~C. \surnamestart Bradfield\surnameend} \&
  \bibinfo{author}{Igor \surnamestart Walukiewicz\surnameend}
  (\bibinfo{year}{2018}): \emph{\bibinfo{title}{The Mu-Calculus and Model
  Checking}}.
\newblock In \bibinfo{editor}{Edmund~M. \surnamestart Clarke\surnameend},
  \bibinfo{editor}{Thomas~A. \surnamestart Henzinger\surnameend},
  \bibinfo{editor}{Helmut \surnamestart Veith\surnameend} \&
  \bibinfo{editor}{Roderick \surnamestart Bloem\surnameend}, editors: {\slshape
  \bibinfo{booktitle}{Handbook of {{Model Checking}}}},
  \bibinfo{publisher}{{Springer International Publishing}},
  \bibinfo{address}{{Cham}}, pp. \bibinfo{pages}{871--919},
  \doi{10.1007/978-3-319-10575-8_26}.

\bibitemdeclare{incollection}{carnielli2020combining}
\bibitem{carnielli2020combining}
\bibinfo{author}{Walter \surnamestart Carnielli\surnameend} \&
  \bibinfo{author}{Marcelo~Esteban \surnamestart Coniglio\surnameend}
  (\bibinfo{year}{2020}): \emph{\bibinfo{title}{Combining {{Logics}}}}.
\newblock In \bibinfo{editor}{Edward~N. \surnamestart Zalta\surnameend},
  editor: {\slshape \bibinfo{booktitle}{The {{Stanford Encyclopedia}} of
  {{Philosophy}}}}, \bibinfo{edition}{fall 2020} edition,
  \bibinfo{publisher}{{Metaphysics Research Lab, Stanford University}}.

\bibitemdeclare{article}{fischerservi1978finite}
\bibitem{fischerservi1978finite}
\bibinfo{author}{Gis{\`e}le \surnamestart Fischer~Servi\surnameend}
  (\bibinfo{year}{1978}): \emph{\bibinfo{title}{The Finite Model Property for
  {{MIPQ}} and Some Consequences}}.
\newblock {\slshape \bibinfo{journal}{Notre Dame Journal of Formal Logic}}
  \bibinfo{volume}{XIX}(\bibinfo{number}{4}), pp. \bibinfo{pages}{687--692}.

\bibitemdeclare{article}{ignatiev1993strong}
\bibitem{ignatiev1993strong}
\bibinfo{author}{Konstantin~N. \surnamestart Ignatiev\surnameend}
  (\bibinfo{year}{1993}): \emph{\bibinfo{title}{On Strong Provability
  Predicates and the Associated Modal Logics}}.
\newblock {\slshape \bibinfo{journal}{Journal of Symbolic Logic}}
  \bibinfo{volume}{58}(\bibinfo{number}{1}), pp. \bibinfo{pages}{249--290},
  \doi{10.2307/2275337}.

\bibitemdeclare{article}{kupke2021sizematters}
\bibitem{kupke2021sizematters}
\bibinfo{author}{Clemens \surnamestart Kupke\surnameend},
  \bibinfo{author}{Johannes \surnamestart Marti\surnameend} \&
  \bibinfo{author}{Yde \surnamestart Venema\surnameend} (\bibinfo{year}{2021}):
  \emph{\bibinfo{title}{On the Size of Disjunctive Formulas in the
  {$\mu$}-Calculus}}.
\newblock {\slshape \bibinfo{journal}{Electronic Proceedings in Theoretical
  Computer Science}} \bibinfo{volume}{346}, pp. \bibinfo{pages}{291--307},
  \doi{10.4204/EPTCS.346.19}.

\bibitemdeclare{incollection}{kurucz2007combining}
\bibitem{kurucz2007combining}
\bibinfo{author}{Agi \surnamestart Kurucz\surnameend} (\bibinfo{year}{2007}):
  \emph{\bibinfo{title}{Combining Modal Logics}}.
\newblock In \bibinfo{editor}{Patrick \surnamestart Blackburn\surnameend},
  \bibinfo{editor}{Johan \surnamestart Van~Benthem\surnameend} \&
  \bibinfo{editor}{Frank \surnamestart Wolter\surnameend}, editors: {\slshape
  \bibinfo{booktitle}{Studies in {{Logic}} and {{Practical Reasoning}}}},
  {\slshape \bibinfo{series}{Handbook of {{Modal Logic}}}}~\bibinfo{volume}{3},
  \bibinfo{publisher}{{Elsevier}}, pp. \bibinfo{pages}{869--924},
  \doi{10.1016/S1570-2464(07)80018-8}.

\bibitemdeclare{article}{ono1977intuitionistic}
\bibitem{ono1977intuitionistic}
\bibinfo{author}{Hiroakira \surnamestart Ono\surnameend}
  (\bibinfo{year}{1977}): \emph{\bibinfo{title}{On Some Intuitionistic Modal
  Logics}}.
\newblock {\slshape \bibinfo{journal}{Publications of the Research Institute
  for Mathematical Sciences}} \bibinfo{volume}{13}(\bibinfo{number}{3}), pp.
  \bibinfo{pages}{687--722}, \doi{10.2977/prims/1195189604}.

\bibitemdeclare{phdthesis}{pacheco2023exploring}
\bibitem{pacheco2023exploring}
\bibinfo{author}{Leonardo \surnamestart Pacheco\surnameend}
  (\bibinfo{year}{2023}): \emph{\bibinfo{title}{Exploring the Difference
  Hierarchies on $\mu$-Calculus and Arithmetic---from the Point of View of
  {{Gale--Stewart}} Games}}.
\newblock Ph.D. thesis, \bibinfo{school}{Tohoku University}.

\bibitemdeclare{misc}{pacheco2023game}
\bibitem{pacheco2023game}
\bibinfo{author}{Leonardo \surnamestart Pacheco\surnameend}
  (\bibinfo{year}{2023}): \emph{\bibinfo{title}{Game Semantics for the
  Constructive $\mu$-Calculus}}, \doi{10.48550/arXiv.2308.16697}.
\newblock \eprint{2308.16697}.

\end{thebibliography}
\end{document}